\documentclass[11pt]{article}
\usepackage{amsmath, amsthm, amscd, amsfonts, amssymb, graphicx, color}
\usepackage[utf8]{inputenc}
\usepackage{graphicx}
\usepackage{hyperref}
\usepackage[titlenumbered,ruled,noend,algosection]{algorithm2e}

\newtheorem{thm}{Theorem}[section]
\newtheorem{cor}[thm]{Corollary}
\newtheorem{lem}[thm]{Lemma}
\newtheorem{prop}[thm]{Proposition}
\newtheorem{defn}[thm]{Definition}

\begin{document}
\title{Approximate Curve-Restricted Simplification of Polygonal Curves}
\author{Ali Gholami Rudi\thanks{
{Department of Electrical and Computer Engineering},
{Bobol Noshirvani University of Technology}, {Babol, Iran}.
Email: {\tt gholamirudi@nit.ac.ir}.}}
\date{}
\maketitle
\begin{abstract}
The goal in the min-\# curve simplification problem is to reduce the number of the
vertices of a polygonal curve without changing its shape significantly.
We study curve-restricted min-\# simplification of
polygonal curves, in which the vertices of the simplified
curve can be placed on any point of the input curve, provided
that they respect the order along that curve.
For local directed Hausdorff distance from the input
to the simplified curve in $\mathbb{R}^2$, we present an
approximation algorithm that computes a curve whose number
of links is at most twice the minimum possible.
\end{abstract}


\section{Introduction}
\label{sint}
The goal of the classical curve simplification problem is to reduce
the number of the vertices of a polygonal curve, without changing
its shape significantly.  There are several applications in which
curve simplification plays an important role.  In trajectory analysis, for instance,
there are two important reasons for this reduction.  First,
it reduces the storage and bandwidth requirements for storing
and transferring huge and growing collections of trajectory data.
Second, and probably more importantly, the complexity of most
trajectory analysis algorithms depends on the number of the
vertices of the input curves, and simplifying trajectories
can reduce the running time of these algorithms.

Let $P = \left< p_1, p_2, ..., p_n \right>$ be a polygonal curve
on the plane.
The curve $P' = \left< p'_1, p'_2, ..., p'_m \right>$ is a
simplification of $P$, if $p'_1 = p_1$, $p'_m = p_n$, $m \le n$,
and the distance between $P$ and $P'$ is at most $\epsilon$ (the
definition of the distance between these curves and the value
of $\epsilon$ is described below).
The simplified curve may be vertex-restricted, curve-restricted, or
unrestricted.  In vertex-restricted simplification, the vertices
of $P'$ coincide with the vertices of the input curve,
i.e.~for each $i$ where $1 \le i \le m$, $p'_i = p_j$ for
some index $j$, where $1 \le j \le n$.
In curve-restricted simplification, the vertices of $P'$ can be
placed on any point of the input curve, and
in unrestricted simplification there is no limitation on the
placement of the internal vertices of $P'$.
In the first two cases, which is the focus of the present paper,
the vertices of the simplified curve should
appear in order on the input curve, and thus split $P$ into sub-curves.
For each edge of the simplification $p'_ip'_{i+1}$,
in which $1 \le i \le m - 1$,
let $P_{p'_ip'_{i+1}}$ denote the sub-curve of $P$ from $p'_i$ to $p'_{i+1}$.

The distance between two curves is computed using measures
such as Fr{\'e}chet or Hausdorff \cite{agarwal05} (other
measures too are sometimes used such as \cite{buzer07}).
Let $D(C, C')$ denote the function that computes the distance
between two curves using any such measure.
The distance between the original and simplified curves is either
\emph{global} and computed for the curves as a whole,
or is \emph{local} and computed as the maximum distance of the
corresponding sub-curves, i.e.~$\max_{1 \le i \le m-1} D(p'_i p'_{i+1}, P_{p'_ip'_{i+1}})$.

Curve simplification is usually studied in two settings \cite{imai86}.
In the min-$\epsilon$ setting the maximum value of $m$ (the number of
the vertices of the simplified curve) is specified and
$\epsilon$ (the amount of distance between the original and
simplified curves) is minimised, and
in the min-\# setting $\epsilon$ is given while $m$ is minimised.
There are numerous results on vertex-restricted curve simplification in
the min-\# setting, only some of which provide a guarantee on the
number of the vertices of the simplification.
In the rest of this paper we focus on the min-\# problem,
and assume that $\epsilon$ is specified as an input.

The well-known algorithm presented by Douglas and Peucker \cite{douglas73}
does not minimise the number of the vertices of the simplified curve,
but is both simple and effective.
It assumes local directed Hausdorff distance from the input curve
to the simplified curve.
For simplifying $P$ with the maximum distance $\epsilon$,
it finds the most distant vertex $p_k$
from segment $p_1 p_n$; if their distance is at most $\epsilon$,
this segment is a link of the simplification.
Otherwise, the algorithm recursively simplifies
$\left<p_1, ..., p_k\right>$ and $\left<p_k, ..., p_n\right>$.
The worst-case time complexity of this algorithm is $O(n^2)$.
Hershberger and Snoeyink~\cite{hershberger94,hershberger97} improved the running
time of this algorithm to $O(n \log n)$ and later to $O(n \log^* n)$.

Among algorithms that compute an optimal simplification,
i.e.~a simplification with the minimum number of links,
the one presented by Imai and Iri~\cite{imai88}
is probably the most popular for local Hausdorff distance.
It creates a shortcut graph, the vertices of which represent
the vertices of the input curve.
An edge $p_ip_j$ shows that the distance between link $p_ip_j$
and sub-curve $\left<p_i, p_{i+1}, ..., p_j\right>$ is at most $\epsilon$.
A shortest path algorithm on this graph, finds the simplification
with the minimum number of vertices.
The time complexity of this algorithm is $O(n^2 \log n)$.
Chan and Chin \cite{chan96}, and also Melkman and O'Rourke \cite{melkman88}
improved the running time of this
algorithm to $O(n^2)$, and Chen and Daescu~\cite{chen03} reduced
its space complexity to $O(n)$.

There are many other results on vertex-restricted simplification
that consider the Fr{\'e}chet distance or compute the distance
of the curves globally.
For instance, van Kreveld at al.~\cite{vankreveld18}
studied the performance of the Douglas and Peucker~\cite{douglas73}
and Imai and Iri~\cite{imai88} algorithms, described above,
under the global Hausdorff or Fr{\'e}chet distance measures.
They showed that computing an optimal vertex-restricted simplification
using the global undirected Hausdorff distance or
global directed Hausdorff distance from the simplified to the
optimal curve is NP-hard, and
presented an output-sensitive dynamic programming algorithm with
the time complexity $O(mn^5)$ for computing an optimal
simplification under the global Fr{\'e}chet distance.
A faster dynamic programming algorithm for the same variation
of the problem was presented by van de Kerkhof et al.~\cite{kerkhof18}
with the time complexity $O(n^4)$.

Some results on vertex-restricted simplification
do not obtain an optimal simplification but
provide a guarantee on the number of the links of the
resulting simplifications using approximation algorithms.
Agarwal et al.~\cite{agarwal05}
for instance, presented a near-linear time approximation algorithm
for local Hausdorff distance using the uniform distance metric,
in which the distance between a point and a curve is defined
as their vertical distance.
They also presented an approximation algorithm for local
Fr{\'e}chet distance under $L_p$ metric.
Both of these algorithms are simple and greedy in nature.
Among these results, there are also vertex-restricted simplification
algorithms that assume streaming input or online
setting \cite{abam10,lin17,cao17,muckell14},
in which a limited storage is available or
the curve should be simplified in one pass.
It is beyond the scope of this paper to review the
literature on curve simplification extensively;
even many heuristic algorithms, such as \cite{chen12,long13},
have been presented for curve simplification (Zhang et al.~\cite{zhang18}
surveyed many of them for trajectory simplification).

Despite the number of results on vertex-restricted curve simplification,
curve-restricted simplification, which has attracted less attention,
can yield a curve with much fewer vertices,
as in Figure~\ref{frestriction}, in which a curve-restricted
simplification with only four vertices is demonstrated for a
curve whose vertex-restricted simplification is the same as the input curve.
For global directed Hausdorff distance, van de Kerkhof et al.~\cite{kerkhof18}
showed that curve-restricted simplification is NP-hard and provided an $O(n)$
algorithm for global Fr{\'e}chet distance in $\mathbb{R}^1$.

\begin{figure}
	\centering
	\includegraphics[width=\columnwidth]{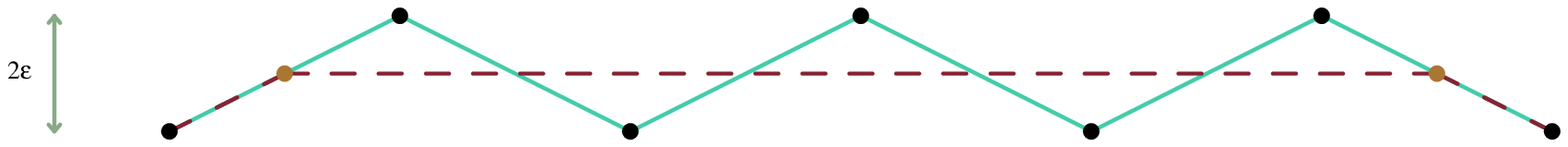}
	\caption{An example showing that curve-restricted simplifications
	can have far fewer vertices compared to vertex-restricted
	simplifications; the dashed links are a curve-restricted
	simplification of the curve with solid edges.}
	\label{frestriction}
\end{figure}

In this paper, we study the min-\# curve-restricted
simplification problem with maximum local Hausdorff
distance $\epsilon$ from the input curve to the simplified curve.
We present a dynamic programming algorithm that computes a simplified curve,
the number of the links of which is at most twice the minimum possible.
This paper is organized as follows: In Section~\ref{sprel}
we introduce the notation used in this paper.
In Section~\ref{slink}, we show how to compute a simplification
link between two edges of the input curve and in Section~\ref{smain},
we present our main algorithm.
We conclude this paper Section~\ref{sconclude}.

\section{Preliminaries and Notation}
\label{sprel}
A two-dimensional polygonal curve is represented as a sequence of vertices
on the plane, with line segments as edges between contiguous vertices.
The directed Hausdorff distance between curves $C$ and $C'$, denoted
as $H(C, C')$, is defined as the maximum of the distance between any
point of $C$ to the curve $C'$, i.e.~$H(C, C') = \max_{p \in C} \mathit{dist}(p, C')$,
in which $\mathit{dist}(p, C')$ is the
Euclidean distance between point $p$ and curve $C'$.

Let $P' = \left< p'_1, p'_2, ..., p'_m \right>$ be a
curve-restricted simplification of $P$.
We have $p'_1 = p_1$, $p'_m = p_n$, $m \le n$,
and the distance between $P$ and $P'$ is at most $\epsilon$.
Also, the vertices of $P'$ should appear in order along $P$.
Given a parameter $\epsilon$, the goal in the min-\# simplification is
to find a simplified curve with the minimum number of vertices, such that the
distance between the original and simplified curves is at most $\epsilon$.
In what follows, we use the term \emph{link} to refer to the edges of
the simplified curve, to distinguish them from the edges of the input curve.

For a link $\ell$ of $P'$, suppose $x$ and $y$ on $P$ are points
corresponding to the start and end points of $\ell$ and suppose $x$ is
on edge $p_ip_{i+1}$ and $y$ is on $p_jp_{j+1}$.
Then, $\ell$ \emph{covers} all edges $p_kp_{k+1}$ for $i \le k \le j$.
Let $P_\ell$ be the sub-curve of $P$ corresponding to link $\ell$,
i.e.~the sub-curve of $P$ from point $x$ to $y$.
The local Hausdorff distance from $P$ to $P'$ is the maximum
of $H(P_\ell, \ell)$ over all links $\ell$ of $P'$.
In this paper we assume local Hausdorff distance
to measure the distance between the input and simplified curves.

The $\epsilon$-neighbourhood of a vertex of $P$ or a segment
which is defined as follows.

\begin{figure}
	\centering
	\includegraphics[width=10cm]{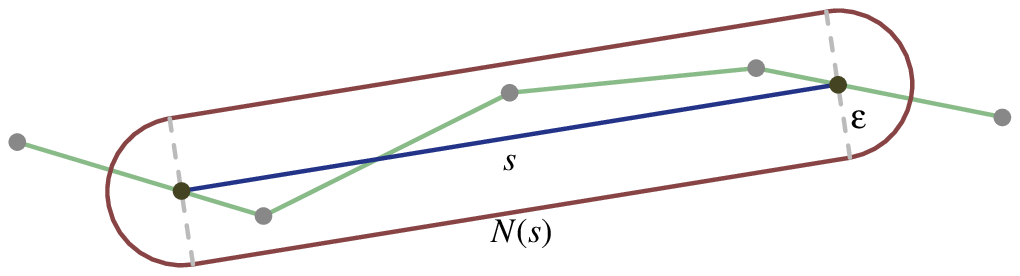}
	\caption{The $\epsilon$ neighbourhood of a segment}
	\label{fneigh}
\end{figure}

\begin{defn}
The $\epsilon$-neighbourhood of a point $p$, denoted as $N(p)$
is a disk of radius $\epsilon$ and with centre is at $p$.
Clearly, the set of points inside $N(p)$ are all points at
distance at most $\epsilon$ from $p$.
Similarly, the $\epsilon$-neighbourhood of a segment $s$,
denoted as $N(s)$, is the set of points at distance at most $\epsilon$
from any point of the segment $s$.
\end{defn}

The $\epsilon$-neighbourhood of a segment $s$ is demonstrated in
Figure~\ref{fneigh}.

\section{Identifying Simplification Links}
\label{slink}

\begin{lem}
\label{llink}
For the curve $P = \left< p_1, p_2, ..., p_n \right>$,
a segment $s$ from point $x$ on edge $p_i p_{i+1}$ to
point $y$ on edge $p_j p_{j+1}$ can be a link of a (not
necessarily optimal) curve-restricted simplification
if and only if it intersects $N(p_k)$ for every index
$k$, where $i < k \le j$.
\end{lem}
\begin{proof}
Let $C$ be the sub-curve $P$ from $x$ to $y$.
If $s$ is a link of a simplification of $P$,
$H(C, s)$ is at most $\epsilon$.
This implies that the distance of every point of $C$ to $s$
is at most $\epsilon$.  For each vertex $p$ of $C$
this means that $s$ should include at least one point
from $N(p)$.

For the converse, suppose $s$ intersects $p_i p_{i+1}$
at $x$ and $p_j p_{j+1}$ at $y$, as well as $N(p)$ for
every vertex of $C$, the sub-curve of $P$ from $x$ to $y$.
It is enough to show that $H(C, s) \le \epsilon$.
For each edge, since the distance between its end points
and $s$ is at most $\epsilon$, the distance of other
points of the edge cannot be greater.  This holds for
every internal edge of $C$ and implies
$H(C, s) \le \epsilon$ as required.
\end{proof}

Lemma~\ref{llink} corresponds to a similar statement for
vertex-restricted simplifications.
We use this lemma later to compute the links of a simplification.

\begin{cor}
\label{clink}
For the curve $P = \left< p_1, p_2, ..., p_n \right>$,
a segment $s$ from point $x$ on edge $p_i p_{i+1}$ to
point $y$ on edge $p_j p_{j+1}$ is a link of a (not
necessarily optimal) curve-restricted simplification
if and only if $N(s)$ contains $p_k$ for every index
$k$, where $i < k \le j$.
\end{cor}

Corollary~\ref{clink} holds because if a segment $s$ intersects
the $\epsilon$-neighbourhood of a vertex $v_k$, the distance of
$v_k$ to $s$ is at most $\epsilon$ and it should be inside $N(s)$.
We use Corollary~\ref{clink} later to improve the time complexity
of detecting simplification links.

\begin{lem}
\label{llinkalgn}
Suppose $\ell$ is a link of a curve-restricted simplification
of curve $P = \left< p_1, p_2, ..., p_n \right>$, such that
$\ell$ starts from point $x$ on edge $p_i p_{i+1}$ and
ends at point $y$ on edge $p_j p_{j+1}$.
There exists another link $\ell'$ covering the same set of edges such
that the line that results from extending $\ell'$ has the following
property for at least two values of $k$ where $i < k \le j$:
either i) it is a tangent to $N(p_k)$, or
ii) it passes through one of the end points of $p_i p_{i+1}$ or
$p_jp_{j+1}$, or their intersection with $N(p_k)$.
\end{lem}
\begin{proof}
Let $L$ be the line resulting from extending the segment $\ell$.
If none of the mentioned properties hold for any value of $k$,
we move $L$ downwards until one of them holds for some value $k$,
i.e.~it becomes tangent to the $\epsilon$-neighbourhood
of $p_k$ or passes through the intersection of the $\epsilon$-neighbourhood
of $p_k$ and the last or the first edge covered by the $s$.
We then rotate $L$ around $p_k$ for case i, or the intersection of
case ii, until one of the conditions holds for another index.
Let $s$ be the segment on line $L$ with end points on $p_ip_{i+1}$
and $p_jp_{j+1}$; such a segment surely exist, since the movement
or rotation stops at the end points of these edges.
\begin{figure}
	\centering
	\includegraphics[width=\columnwidth]{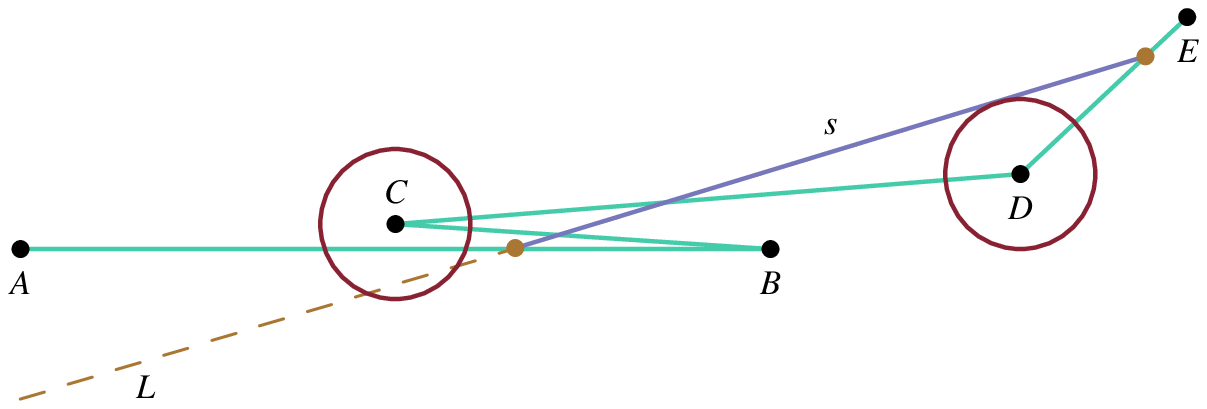}
	\caption{Rotating line $L$ around $N(D)$ counterclockwise;
	$s$ no longer intersects $N(C)$.}
	\label{fshort}
\end{figure}

Clearly $L$ cannot leave $N(p_k)$ for any possible index $k$ for
both the downward movement and the rotation; just before leaving
$N(p_k)$, $L$ becomes its tangent.
The only problem may be that although $N(p_k)$, for some $k$
where $i < k \le j$, is intersected by both $\ell$ and $L$,
$s$ may be too short to intersect $N(p_k)$; this is demonstrated
in Figure~\ref{fshort}.  However, since the rotation stops
at the intersection the first or the last edge and $N(p_k)$,
this case never happens.
\end{proof}

\begin{lem}
\label{llinkfind}
A link of a curve-restricted simplification
of $P = \left< p_1, p_2, ..., p_n \right>$,
from a point on edge $p_i p_{i+1}$ to a point on edge $p_j p_{j+1}$
can be found with the time complexity $O(m^3)$ where $m = j - i + 1$,
provided such a link exists.
\end{lem}
\begin{proof}
We find a line for which the condition mentioned in Lemma~\ref{llinkalgn}
holds.  To do so, we find three parallel lines at
distance $\epsilon$ on the plane, $L_1$, $L_2$, and $L_3$, such
that a link can be found on line $L_2$.
We consider possible placements of these lines according to Lemma~\ref{llinkalgn} and
check for which of them the condition of Lemma~\ref{llink} holds for a segment on $L_2$.
If $L_2$ is a tangent to $N(p_k)$ for some value of $k$ where
$i < k \le j$, then either $L_1$ or $L_3$ should pass through
$p_k$.  We therefore try different placements of these three lines
such that the following property holds for two values of $k$ for $i < k \le j$:
either i) $L_1$ or $L_3$ passes through $p_k$,
or ii) $L_2$ passes through the intersection $N(p_k)$ and
one of $p_{i - 1}p_i$ or $p_jp_{j+1}$ or the end points of these edges.
Since there are $O(m)$ choices for the first and the second
conditions, the number of total cases to consider is $O(m^2)$.

For each of $O(n^2)$ possible placements of these lines,
we have to verify if there exists
a segment $s$ on $L_2$ such that $H(C, s)$ is at most $\epsilon$.
Let $x$ be the intersection of $L_2$ and $p_i p_{i+1}$ and
let $y$ be the intersection of $L_2$ and $p_j p_{j+1}$; if
$x$ or $y$ do not exist, $L_2$ cannot contain a link.
Based on Lemma~\ref{llink}, if the segment $xy$ intersects
$N(p_k)$ for every $i < k \le j$, it is a valid link.
This can be checked with the time complexity $O(m)$.
\end{proof}

\begin{cor}
\label{rlinkfind}
To force the link to start from $p_i$, instead of
any point on edge $p_ip_{i+1}$ in Lemma~\ref{llinkfind},
we can fix this point on $L_2$ and try the condition mentioned
in the proof of Lemma~\ref{llinkfind} for only one value of $k$.
\end{cor}

Algorithms based on the construction of the shortcut graph of
Imai and Iri~\cite{imai88} perform steps similar to Lemma~\ref{llinkfind}:
for each $i$ and $j$, where $1 \le i < j \le n$, it should
be verified if the segment $p_i p_j$ intersects the
$\epsilon$-neighbourhood of every vertex $p_k$ for $i < k < j$.
This task can be optimised by computing
the set of lines that pass through $p_i$ and intersect the
$\epsilon$-neighbourhood of the vertices that appear after it
(the intersection of double cones; see \cite{chen03}, for instance).
Unfortunately, for curve-restricted simplification that does
not seem possible, since the end points of each link may not be
a vertex and are chosen from a much larger set (see Lemma~\ref{llinkalgn}).
Therefore, to improve the time complexity of Lemma~\ref{llinkfind},
we should use an alternative strategy.

\begin{figure}
	\centering
	\includegraphics[width=8cm]{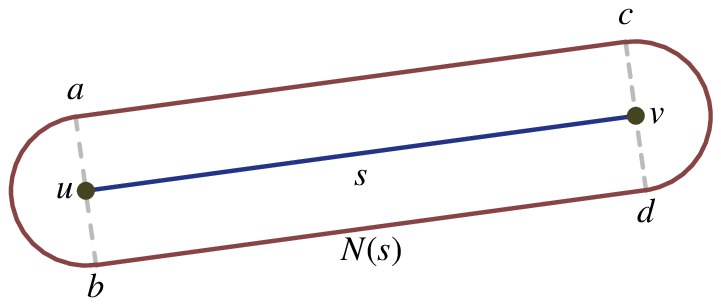}
	\caption{Symbols used for $N(s)$ in Lemma~\ref{llinkfast}}
	\label{flinkfast}
\end{figure}

\begin{lem}
\label{llinkfast}
Let $S$ be a set of $n$ points on the plane and let $\delta$ be a constant,
where $0 < \delta < 1$.
There exists a data structure with $O(n^{1+\delta})$ preprocessing time and space,
which, for any segment $s$, can verify if all points in $S$ are
inside the $\epsilon$-neighbourhood of $s$ in $O(2^{1/\delta}\log n)$ time.
\end{lem}
\begin{proof}
We first compute the convex hull $H$ of the points in $S$.
The most distant point of $S$ from $s$ is a vertex of $H$.
Let $\ell(x, y)$ be the line that results from extending
the segment from point $x$ to point $y$, and
let $h(x, y)$ be the halfplane on the left side of $\ell(x, y)$.
All members of $S$ are in $N(s)$, if and only if there is no point
in the following four regions (we use the symbols defined in Figure~\ref{flinkfast}):
\begin{enumerate}
\item $h(a, c)$,
\item $h(d, b)$,
\item $h(b, a) \setminus N(u)$, and
\item $h(c, d) \setminus N(u)$.
\end{enumerate}
Since, the intersections of a convex polygon and a line can be
computed in logarithmic time, the first two regions can be checked
in $O(\log n)$.
The other two regions can be checked using \emph{halfplane proximity queries}:
given a directed line $\ell$ and a point $q$, report the point farthest from $q$
among those to the left of $\ell$.
Aronov et al.~\cite{aronov18} presented a data structure that uses $O(n^{1 + \delta})$
preprocessing time and space, to answer such queries in $O(2^{1 / \delta} \log n)$ time,
for any $\delta$ ($0 < \delta < 1$).
Therefore, to check the third region, we perform a halfplane proximity query
for line $\ell(b, a)$ and point $u$; only if the distance of the farthest point
to $u$ in $h(b, a)$ is at most $\epsilon$, the third region is empty.
Similarly, to check the fourth region, we perform a halfplane proximity query,
specifying line $\ell(c, d)$ and point $v$ as inputs.
\end{proof}

\begin{lem}
\label{llinkfindfast}
Let $\delta$ be a constant, where $0 < \delta < 1$.
With $O(n^{3+\delta})$ preprocessing time and space,
a link of a curve-restricted simplification
of a polygonal curve $P = \left< p_1, p_2, ..., p_n \right>$,
from any edge $p_i p_{i+1}$ to any other edge $p_j p_{j+1}$
can be found with the time complexity $O(n^2 \log n)$,
provided that such a link exists.
\end{lem}
\begin{proof}
For every pair of indices $x$ and $y$, where $1 < x \le y < n$,
we initialize the data structure mentioned in Lemma~\ref{llinkfast}
$D_{xy}$ for points $\{ p_x, p_{x + 1}, ..., p_y \}$.
This can be done with the time complexity $O(n^{3 + \delta})$.
In Lemma~\ref{llinkfind}, to check if a segment from $p_i p_{i+1}$
to $p_j p_{j+1}$ is a link, we test to see if it intersects $N(p_k)$
for every index $k$, where $i < k \le j$.  We improve the time
complexity of this task to $O(\log n)$ by using $D_{xy}$.
\end{proof}

\section{Simplification Algorithm}
\label{smain}

\begin{figure}
	\centering
	\includegraphics[width=\columnwidth]{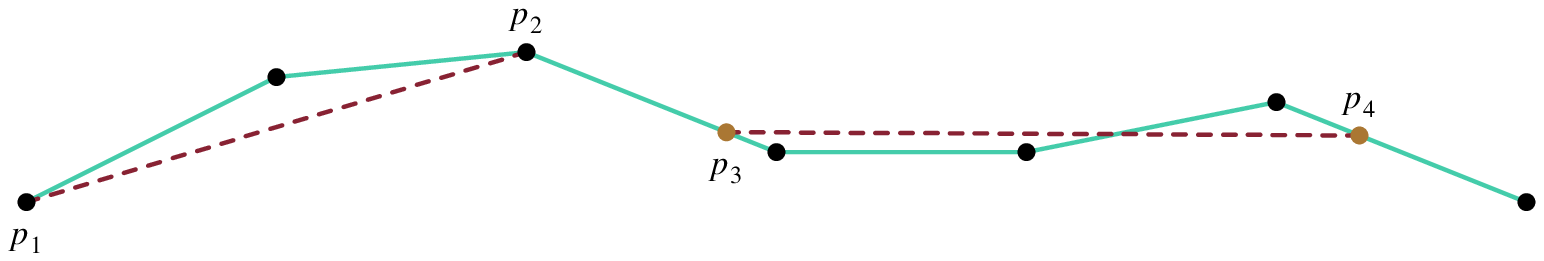}
	\caption{A DLC $\left< p'_1 p'_2, p'_3 p'_4 \right>$
	of a curve with six links (Definition~\ref{ddlc})}
	\label{fdlc}
\end{figure}

\begin{defn}
\label{ddlc}
A sequence of segments $D = \left< p'_1 p'_2, p'_3 p'_4, ..., p'_{2k - 1} p'_{2k} \right>$
is a \emph{disjoint link chain} (DLC) for curve
$P = \left< p_1, p_2, ..., p_n \right>$, if
i) $p'_1$ is on $p_1 p_2$ and $p'_{2k}$ is on $p_{n-1} p_n$,
ii) for each index $i$,
where $1 \le i \le k$, $p'_{2i - 1} p'_{2i}$ is a valid simplification link, and
iii) for each index $i$,
where $1 \le i < k$, $p'_{2i}$ and $p'_{2i+1}$ are on the same edge of $P$,
and
iv) the vertices of $D$ appear in order on $P$ (i.e.~first $p'_1$
appears on $P$, then $p'_2$, then $p'_3$, and so forth).
\end{defn}

Figure~\ref{fdlc} demonstrates a DLC of a curve with six links.

\begin{prop}
\label{rdlc}
Given a DLC $D = \left< p'_1 p'_2, p'_3 p'_4, ..., p'_{2k - 1} p'_{2k} \right>$
for curve $P = \left< p_1, p_2, ..., p_n \right>$,
such that $p'_1 = p_1$,
a curve-restricted simplification of $P$ with $2k$ links can be obtained
from $D$ by connecting the end of each link of $D$ to the start of its next
link and connecting the end of the last one to $p_n$.
\end{prop}

\begin{figure}
	\centering
	\includegraphics[width=\columnwidth]{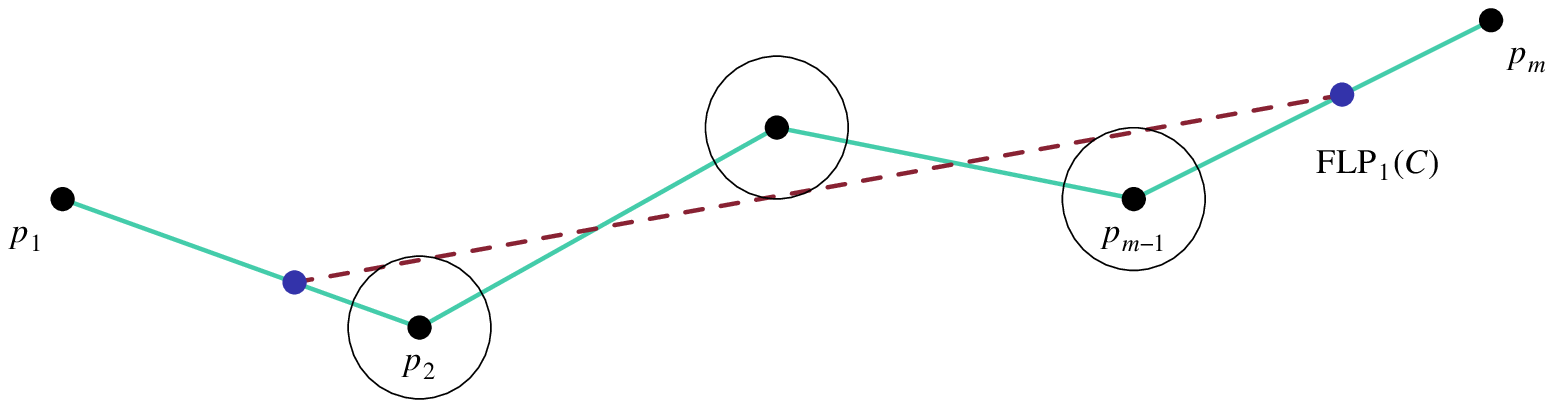}
	\caption{$\mathrm{FLP}_1(\left< p_1, p_2, ..., p_m \right>)$
	of a curve with six links (Definition~\ref{dflp})}
	\label{fflp}
\end{figure}

\begin{defn}
\label{dflp}
For a polygonal curve $C = \left< p_1, p_2, ..., p_m \right>$,
the \emph{first link point} with $k$ links,
denoted as $\mathrm{FLP}_k(C)$,
is the first point $x$ on $p_{m-1} p_m$,
such that there exists a disjoint link
chain $D= \left< p'_1 p'_2, p'_3 p'_4, ..., p'_{2k - 1} p'_{2k} \right>$ of $C$,
in which $p'_{2k} = x$.
\end{defn}

Figure~\ref{fflp} demonstrates $\mathrm{FLP}_1$ of a curve with four edges.
Since the line containing a link can be moved or rotated
to obtain a new link, unless the conditions mentioned in
Lemma~\ref{llinkalgn} holds for it,
Lemma~\ref{llinkfindfast} yields the following corollary.

\begin{cor}
\label{ccorner}
For a sub-curve $Q = \left< q_1, q_2, ..., q_m \right>$
of a polygonal curve $P = \left< p_1, p_2, ..., p_n \right>$,
$\mathrm{FLP}_1(Q)$ and its corresponding link can be computed
with the time complexity $O(n^2 \log n)$,
after some preprocessing with the time complexity $O(n^{3+\delta})$,
for some constant $\delta$ ($0 < \delta < 1$).
\end{cor}

In Theorem~\ref{tmain} we present an algorithm for computing a
minimum-sized DLC.

\begin{thm}
\label{tmain}
A DLC of minimum size for curve
$P = \left< p_1, p_2, ..., p_n \right>$
can be computed in $O(n^5 \log n)$.
\end{thm}
\begin{proof}
We use dynamic programming to fill a two-dimensional table $F$.
$F[i, j]$, for $1 \le i \le n$ and $1 \le j \le n$, denotes
$\mathrm{FLP}_j(\left< p_1, p_2, ..., p_i \right>)$.
Parallel to table $F$, we can store the last link of $F[i, j]$
in another two-dimensional table $L$ to reconstruct the chain.
For points $u$ and $v$ on $P$, $u < v$ holds if $u$ appears
before $v$ on $P$.  We fill the tables as follows.
\begin{enumerate}
\item $F[i][1]$ is initialized as $\mathrm{FLP}_1(\left<p_1, p_2, ..., p_i\right>)$,
forcing the first vertex of the resulting link to be on $p_1$ (Corollary~\ref{rlinkfind}).
$L[i][1]$ is initialised as the link corresponding to $\mathrm{FLP}_1(\left<p_1, p_2, ..., p_i\right>)$.
If there is no such link, $F[i][1]$ and $L[i][1]$ are not filled.
\item
For $d$ from $2$ to $n$, $F[i][d]$ and $L[i][d]$ for $1 \le i \le n$
are filled as follows:
The value $F[i][d]$ is the minimum value
of $\mathrm{FLT}_1(\left< F[j][d - 1], p_{j+1}, p_{j+2}, ..., p_{i}\right>)$,
over all indices of $j$, where $j < i$ and $F[j][d - 1]$ is filled.
The value of $L[i][d]$ should indicate the link corresponding to
of $\mathrm{FLT}_1(\left< F[j][d - 1], p_{j+1}, p_{j+2}, ..., p_{i}\right>)$.
\end{enumerate}
Based on Corollary~\ref{ccorner}, filling these tables can be done with
the time complexity $O(n^5 \log n)$.

Let $m$ be the lowest index, such that $F[n][m]$ is filled.  By
following the links backwards using dynamic programming tables,
we obtain a DLC
$D = \left< p'_1 p'_2, p'_3 p'_4, ..., p'_{2k-1} p'_{2k} \right>$.
We prove that the size of $D$ is the minimum possible.
To do so, we use induction on $d$ to show that $F[i][d]$ for $1 \le i \le n$
is filled if and only if there is a DLC for
$\left< p_1, p_2, ..., p_i \right>$ with $d$ links.
For $d = 1$, the statement is trivial and follows from
the definition of $\mathrm{FLT}_1$ and its computation
(Corollary~\ref{rlinkfind}).
For $d > 1$, suppose there is a DLC
$\left< q'_1 q'_2, q'_3 q'_4, ..., q'_{2d - 1} q'_{2d} \right>$
for $\left< p_1, p_2, ..., p_i \right>$.
Let $q'_{2d - 2}$ be on $p_{j} p_{j + 1}$.
Obviously, $\left< q'_1 q'_2, q'_3 q'_4, ..., q'_{2d - 3} q'_{2d - 2} \right>$
is a DLC of $\left< p_1, p_2, ..., p_{j + 1} \right>$.
By induction hypothesis, $F[j][d - 1]$ is filled with
a point on or before $q'_{2d - d}$.
Since $q'_{2d - 1} q'_{2d}$ is a valid link, where $q'_{2d - 1}$
appears after $q'_{2d - 2}$ on $P$,
there is a valid link from $q'_{2d - 2} p_{j+1}$ to $p_{i - 1} p_i$,
and $P[i][d]$ is filled in the dynamic programming algorithm.
\end{proof}

\begin{thm}
\label{tmaincomp}
A curve-restricted simplification of a polygonal curve
$P = \left< p_1, p_2, ..., p_n \right>$
can be computed in $O(n^5 \log n)$,
such that its number of links is at most twice the number
of the links of an optimal simplification.
\end{thm}
\begin{proof}
Let $D$ be the DLC of $P$ with $k$ links
computed using Theorem~\ref{tmain}.
We can obtain a curve-restricted simplification
$P'$ from $D$ with $m = 2k$ links (Proposition~\ref{rdlc}).
Let $O$ be a curve-restricted simplification of $P$ with
the minimum number of links $x$.
Based on Definition~\ref{ddlc}, $O$ is also a DLC of $P$.
Since $D$ is a DLC with the minimum number of links,
$x \ge k$.  This implies $2x \ge 2k = m$.
\end{proof}

\section{Concluding Remarks}
\label{sconclude}
Although, the min-\# curve-restricted simplification of
polygonal curves can reduce the number of the vertices of the
curves much better than vertex-restricted simplification, the time complexity of the
algorithm presented in this paper is not very appealing for
real-world applications.  A faster approximate or exact algorithm
may fill this gap.

%\bibliographystyle{unsrt}
%\bibliography{sima}

\end{document}